%% file: main.tex
\documentclass[a4paper, amsfonts, amssymb, amsmath, reprint, showkeys, nofootinbib, twoside]{revtex4-1}
\usepackage[english]{babel}
\usepackage[utf8]{inputenc}
\usepackage[colorinlistoftodos, color=green!40, prependcaption]{todonotes}
\input{preamble}
\usepackage[pdftex, pdftitle={Article}, pdfauthor={Author}]{hyperref} 
\usepackage{tikz}
\usetikzlibrary{arrows.meta,calc}
\usepackage{subcaption}
\usepackage{booktabs}
\usepackage{amsfonts}
\usepackage{amsmath}
\usepackage{amsthm}
\usepackage{amssymb}

\newtheorem{corollary}{Corollary}

\newtheorem{proposition}{Proposition}

\begin{document}

\title{Directional Liquidity and Geometric Shear in Pregeometric Order Books}

\author{Jo\~ao P. da Cruz}
 \affiliation{The Quantum Computer Company, Lisbon, Portugal}
  \email{joao@quantumcomp.pt}
 \affiliation{Center for Theoretical and Computational Physics, Lisbon, Portugal}

\date{\today} 

\begin{abstract}
We introduce a structural framework for the geometry of financial order books in
which liquidity, supply, and demand are treated as emergent observables rather
than primitive market variables.
The market is modeled as a relational substrate without assumed metric,
temporal, or price coordinates.
Observable quantities arise only through observation, implemented here as a
reduction of relational degrees of freedom followed by a low-dimensional
spectral projection.
A one-dimensional projection induces a price-like coordinate and a projected
liquidity density around the mid price, from which bid and ask sides emerge as
two complementary restrictions.
We show that directional liquidity imbalances decompose naturally into a rigid
drift of the projected density and a geometric shear mode that deforms the
bid--ask structure without inducing price motion.
Under a minimal single-scale hypothesis, the shear geometry constrains the
projected liquidity to a gamma-like functional form, appearing as an
integrated-gamma profile in discrete data.
Empirical analysis of high-frequency Level~II data across multiple U.S. equities
confirms this geometry and shows that it outperforms standard alternative
cumulative models under explicit model comparison and residual diagnostics.
\end{abstract}


\maketitle
\input{intro}
\input{obs_geometry.tex}  
\input{shear_decomp}

\input{shear_and_gamma}
\input{empirical}
\input{discussion}
\input{future}

\input{acknowledgements.tex}


\end{document}

%% file: preamble.tex
\usepackage{amsthm}
\usepackage{mathtools}
\usepackage{physics}
\usepackage{xcolor}
\usepackage{graphicx}
\usepackage[left=23mm,right=13mm,top=35mm,columnsep=15pt]{geometry} 
\usepackage{adjustbox}
\usepackage{placeins}
\usepackage[T1]{fontenc}
\usepackage{lipsum}
\usepackage{csquotes}

%% file: intro.tex
\section{Introduction}
\label{sec:introduction}

The structure of liquidity in financial order books has been a central object
of study in market microstructure for more than two decades.
Empirical investigations of high-frequency data have revealed remarkably
robust regularities in the shape of order books, including convex liquidity
profiles near the mid price, heavy tails at larger price distances, and
persistent bid--ask asymmetries
\cite{Bouchaud2002,Bouchaud2004,Bouchaud2009,SmithFarmerGillemotKrishnamurthy2003}.
These features appear across assets, venues, and market regimes, suggesting
that they reflect structural constraints rather than idiosyncratic trading
strategies.

Most existing approaches explain these regularities by modeling explicit
microstructural mechanisms.
Agent-based and stochastic models describe order placement, cancellation, and
execution as driven by heterogeneous trader behavior, inventory control, or
strategic optimization
\cite{ZovkoFarmer2002,FarmerLillo2004,MikeFarmer2008,ContStoikovTalreja2010}.
Within this paradigm, bid and ask curves are treated as independent objects,
interpreted as opposing forces of supply and demand whose interaction produces
price dynamics and liquidity fluctuations.

While successful in reproducing many stylized facts, such models typically
require detailed behavioral assumptions and asset-specific calibration.
Moreover, they take price and time as fundamental variables, embedding the
order book in a pre-existing metric and temporal structure.
As a result, it remains unclear whether the observed geometric regularities of
liquidity are consequences of specific trading rules, or whether they reflect
more general constraints imposed by observation itself.

In this work, we adopt a complementary and deliberately minimal perspective.
Rather than postulating price, time, or supply and demand as primitive market
variables, we treat them as \emph{emergent observables}.
The market is modeled as a relational substrate whose microscopic description
contains no metric, no temporal coordinate, and no notion of price.
Observable quantities arise only through projection, implemented here as a
reduction of relational degrees of freedom followed by a low-dimensional
spectral embedding.
This perspective is inspired by pregeometric approaches in statistical physics
and quantum gravity, where geometry and dynamics are not fundamental but emerge
from relational structures under coarse-graining
\cite{Rovelli2004,Oriti2014,AmbjornJurkevichLoll2005}.

Within this framework, the order book is not viewed as a collection of
independent bid and ask curves.
Instead, liquidity is represented as a single projected density defined along
an emergent price-like coordinate.
The bid and ask sides arise only after an observational cut at the mid price,
as two complementary restrictions of the same underlying object.
Bid--ask asymmetry is therefore not attributed to distinct economic forces,
but to geometric deformation of the projected density.

A central theme of the present paper is that directional liquidity imbalance
can be decomposed into two qualitatively different modes.
One corresponds to a rigid translation of the projected density, associated
with mid-price drift.
The other corresponds to a relative deformation between the two sides of the
mid, which we term \emph{geometric shear}.
These modes are conceptually distinct: shear modifies the shape of the
bid--ask structure without inducing price motion, while drift moves the mid
without changing local geometry.
This separation provides a natural explanation for empirical observations
showing that order-book imbalance need not be directly predictive of price
changes \cite{FarmerLillo2004,Bouchaud2009}.

Under minimal regularity assumptions and a single-scale hypothesis excluding
intrinsic length scales beyond distance to the mid and finite visibility, we
show that the projected liquidity density is constrained to a gamma-like
functional form.
In empirical data, where discretization and sparsity obscure the differential
profile, this prediction manifests as an integrated-gamma geometry of
cumulative liquidity.
We validate these predictions using high-frequency Level~II data for several
U.S. equities and show that the proposed geometry outperforms standard
alternative models under explicit model comparison and residual analysis.

The paper is organized as follows.
Section~\ref{sec:obs_geometry} introduces the observational construction,
showing how liquidity, the mid price, and bid--ask structure arise from
projection of a relational substrate.
Section~\ref{sec:shear} develops the decomposition of liquidity dynamics into drift and
shear modes.
Section~\ref{sec:single_scale_shear} derives the single-scale shear constraint and the resulting
gamma-like geometry.
Section~\ref{sec:empirical_shear} presents empirical validation using real order-book data, including
model comparison and residual diagnostics.
We conclude with a discussion of implications for market microstructure and
directions for future work.

%% file: obs_geometry.tex
\section{Observational geometry and two-sided restriction}
\label{sec:obs_geometry}

The central premise of this work is that order-book observables do not exist at
the level of the underlying market substrate, but arise only through
observation.
In particular, price, liquidity, and bid--ask structure are not assumed as
primitive variables.
They emerge as geometric objects induced by a reduction and projection of an
underlying relational system.

Figure~\ref{fig:obs_geometry_vertical} provides a schematic overview of this
construction and will guide the discussion throughout this section.

\input{figura}

\subsection{Relational substrate and observational weights}

We begin with a market substrate represented abstractly as a relational system
$G=(V,E)$, where vertices label economic entities or interaction sites and edges
encode the possibility of interaction.
No metric, temporal ordering, or price coordinate is defined at this level.

Each vertex $i\in V$ carries a nonnegative observational weight $w_i\ge 0$,
representing visible size or intensity.
These weights define a purely atomic measure on the substrate,
\begin{equation}
\mu_G := \sum_{i\in V} w_i\,\delta_i ,
\end{equation}
where $\delta_i$ denotes the unit point mass at vertex $i$.
This measure contains no geometric information; it records only relational
support and magnitude.

Importantly, the weights $w_i$ should already be understood as
\emph{observational} quantities.
They arise after marginalization of underlying relational activity and do not
correspond to microscopic degrees of freedom.

Panel~(a) of Fig.~\ref{fig:obs_geometry_vertical} illustrates this level: a
weighted but non-geometric relational substrate.

\subsection{Edge-level degrees of freedom and observational reduction}
\label{sec:edge_to_node_reduction}

In the present framework, the microscopic degrees of freedom are relational:
they live on \emph{edges} rather than on vertices.
Operationally, an order-book snapshot consists of interactions across venues,
levels, and counterparties, which are naturally represented as edge-level
contributions to liquidity.

The observer does not access these edge degrees of freedom directly.
Instead, observation produces a reduced description by compressing
edge-level information into vertex-level observables.

Formally, let $\mathbb{R}^{E}$ denote an edge space carrying relational activity
variables $u\in\mathbb{R}^{E}$.
A generic reduction to vertex space $\mathbb{R}^{V}$ can be written schematically
as
\begin{equation}
\label{eq:edge_to_node_operator}
x = B\,u ,
\end{equation}
where $B\in\mathbb{R}^{|V|\times|E|}$ is an incidence-type aggregation operator
(e.g., oriented or weighted incidence).
Equation~\eqref{eq:edge_to_node_operator} is not meant as a microstructural
model; it encodes the generic elimination of unobserved relational degrees of
freedom.

This reduction step captures a central pregeometric idea: observable quantities
arise only after compressing the underlying relational structure.
Vertices serve as observable anchors only after edge-level information has been
integrated out.

\subsection{Observational projection}

Observable coordinates arise only through a further projection of the reduced
vertex signal.
We denote by
\begin{equation}
p = \Pi(G)
\end{equation}
an observational projection assigning to each vertex $i$ a real-valued
coordinate $p_i\in\mathbb{R}$.

In practice, such projections may arise, for example, from
low-dimensional spectral representations of graph operators
constructed from relational structure.
However, the arguments below rely only on the existence of a
one-dimensional observable coordinate.

Crucially, the projection induces an \emph{ordering} but not an intrinsic
metric.
Distances are meaningful only relative to the projection and only within a
finite observational window.
Panel~(b) of Fig.~\ref{fig:obs_geometry_vertical} depicts this step: the
relational substrate collapses onto a one-dimensional observable axis.

\subsection{Pushforward liquidity measure}

The observable liquidity distribution is obtained by transporting relational
weights through the projection.
Formally, the projected liquidity measure is defined as the pushforward of
$\mu_G$ under $p$,
\begin{equation}
\nu_G := p_\# \mu_G .
\end{equation}

Operationally, this means that all weights carried by vertices mapping to the
same projected coordinate are aggregated.
In empirical data, $\nu_G$ is accessed only through discretization and finite
windows, yielding a binned liquidity density along the projected coordinate.

At this stage, there is still only \emph{one} observable liquidity object.
No distinction between bid and ask has yet been introduced.
Panel~(c) of Fig.~\ref{fig:obs_geometry_vertical} illustrates this step.

\subsection{Mid price as an observational cut}

The bid--ask structure arises through a further observational operation: a
distinguished cut of the projected axis.
We define the mid price $p^\star$ as the point that balances projected mass on
either side,
\begin{equation}
p^\star :=
\arg\min_p
\left|
\int_{-\infty}^{p} \nu_G(du)
-
\int_{p}^{\infty} \nu_G(du)
\right|.
\end{equation}

This definition is purely geometric.
It involves no notion of transaction, valuation, equilibrium, or market
clearing.
The mid price is the natural symmetry point of the projected measure within the
observational window.

\subsection{Two-sided restriction}

Once the mid price is defined, the observable liquidity splits into two
restricted branches,
\begin{align*}
Q_{\mathrm{bid}}(x) &:= \nu_G(p^\star - x), \\
Q_{\mathrm{ask}}(x) &:= \nu_G(p^\star + x), \\ x&>0 .
\end{align*}

These are not independent curves.
They are complementary restrictions of the \emph{same} projected density.
Panel~(d) of Fig.~\ref{fig:obs_geometry_vertical} illustrates this final step.

This construction makes clear that bid--ask asymmetry does not require
independent forces of supply and demand.
Asymmetry arises whenever the projected density is locally skewed around the
mid, a geometric effect that we later interpret as a directional shear induced
by inflationary relational dynamics.

In the following sections, we analyze the structural consequences of this
two-sided restriction and show that, under minimal regularity assumptions, it
leads uniquely to the integrated-gamma liquidity geometry observed in real
order books.

\subsection{Operational meaning of the projection $\Pi(G)$}

An important clarification concerns the role of the projection
$\Pi(G)$ introduced in this section.
In the abstract framework, $\Pi(G)$ denotes an observational operator
that reduces an underlying relational substrate to a one-dimensional
coordinate accessible to the observer.
This projection is not assumed to be unique, nor is its explicit
construction required for the results derived in this work.

In the empirical analysis, the role of $\Pi(G)$ is played directly by
the observed price axis provided by the Level~II order book data.
That is, the discretized price coordinate (in tick units relative to
the mid) constitutes the operational realization of the projection.
Liquidity is observed only after this reduction, through aggregation
of visible depth at each projected price level.

Accordingly, no explicit spectral embedding or graph Laplacian is
constructed from the data; the observable coordinate is taken directly
from the market feed.
The pregeometric framework should therefore be understood as a
conceptual organization of the observational procedure, rather than as
an additional empirical hypothesis.
Its purpose is to clarify which properties of order-book observables
are structural consequences of projection and restriction, and which
depend on market-specific microstructural details.

This distinction between operational observables and abstract
observational structure aligns the present approach with statistical
physics treatments of coarse-graining and projection, rather than with
microscopic market modeling.

From this perspective, the results reported here apply to any
one-dimensional observable coordinate that orders liquidity and admits
a natural symmetry point.
The price axis is a privileged example, but not the only possible
realization of the projection $\Pi(G)$.

%% file: figura.tex
\begin{figure}[t]
\centering
\begin{tikzpicture}[
    font=\small,
    panel/.style={draw, rounded corners=2pt, minimum width=0.75\columnwidth, minimum height=3.2cm, inner sep=6pt},
    lab/.style={font=\small\bfseries},
    arr/.style={-{Latex[length=3mm]}, line width=1.0pt},
    midline/.style={densely dashed, line width=0.6pt},
    axis/.style={line width=0.8pt},
    stem/.style={line width=0.8pt},
    dot/.style={circle, fill=black, inner sep=1.0pt},
]

\coordinate (cA) at (0,  0.0);
\coordinate (cB) at (0, -4.2);
\coordinate (cC) at (0, -8.4);
\coordinate (cD) at (0,-12.6);

\node[panel] (A) at (cA) {};
\node[lab, anchor=north west] at (A.north west) {(a) Relational substrate};

\coordinate (v2)  at ([xshift=-18mm,yshift= 4mm]A.center);
\coordinate (v3)  at ([xshift=-10mm,yshift= 6mm]A.center);
\coordinate (v4)  at ([xshift= -2mm,yshift= 3mm]A.center);
\coordinate (v5)  at ([xshift=  6mm,yshift= 8mm]A.center);
\coordinate (v6)  at ([xshift= 14mm,yshift= 5mm]A.center);
\coordinate (v7)  at ([xshift= 22mm,yshift= 7mm]A.center);

\coordinate (v9)  at ([xshift=-12mm,yshift=-5mm]A.center);
\coordinate (v10) at ([xshift= -4mm,yshift=-2mm]A.center);
\coordinate (v11) at ([xshift=  4mm,yshift=-4mm]A.center);
\coordinate (v12) at ([xshift= 12mm,yshift=-3mm]A.center);
\coordinate (v13) at ([xshift= 20mm,yshift=-5mm]A.center);

\coordinate (v14) at ([xshift=-10mm,yshift=-10mm]A.center);
\coordinate (v15) at ([xshift=  0mm,yshift=-11mm]A.center);
\coordinate (v16) at ([xshift= 10mm,yshift=-10mm]A.center);

\draw[line width=0.6pt, gray!70] (v3)--(v2);
\draw[line width=0.6pt, gray!70] (v3)--(v4);
\draw[line width=0.6pt, gray!70] (v3)--(v5);
\draw[line width=0.6pt, gray!70] (v3)--(v9);
\draw[line width=0.6pt, gray!70] (v3)--(v10);

\draw[line width=0.6pt, gray!60] (v10)--(v9);
\draw[line width=0.6pt, gray!60] (v10)--(v11);
\draw[line width=0.6pt, gray!60] (v11)--(v12);

\draw[line width=0.6pt, gray!60] (v9)--(v14);
\draw[line width=0.6pt, gray!60] (v10)--(v15);
\draw[line width=0.6pt, gray!60] (v11)--(v16);

\draw[line width=0.6pt, gray!50] (v5)--(v6);
\draw[line width=0.6pt, gray!50] (v6)--(v7);
\draw[line width=0.6pt, gray!50] (v12)--(v13);

\foreach \v/\s in {
  v2/1.1, v3/1.6, v4/1.1, v5/1.4, v6/1.1, v7/1.3,
  v9/1.0, v10/1.3, v11/1.0, v12/1.4, v13/1.0,
  v14/1.1, v15/1.3, v16/1.1
}{
  \node[circle, fill=black, inner sep=\s pt] at (\v) {};
}

\node[font=\scriptsize, anchor=south west]
  at ([xshift=2mm,yshift=1mm]A.south west)
  {vertices $i$, weights $w_i$, edges encode adjacency};

\draw[arr] (A.south) -- ++(0,-8mm)
  node[midway, right=2mm, font=\scriptsize] {projection};

\node[panel] (B) at (cB) {};
\node[lab, anchor=north west] at (B.north west) {(b) Observational projection};

\draw[midline] ([yshift=1.2cm]B.center) -- ([yshift=-1.2cm]B.center);

\foreach \y/\s in {1.1/1.3, 0.8/1.0, 0.5/1.2, 0.2/1.0, 0.0/1.6, -0.3/1.1, -0.7/1.3, -1.1/1.0}{
  \node[circle, fill=black, inner sep=\s pt] at ([yshift=\y cm]B.center) {};
}

\node[font=\scriptsize, anchor=south]
  at ([yshift=0.5mm]B.south)
  {$p = \Pi(G)$};

\draw[arr] (B.south) -- ++(0,-8mm)
  node[midway, right=2mm, font=\scriptsize] {pushforward};

\node[panel] (C) at (cC) {};
\node[lab, anchor=north west] at (C.north west) {(c) Observable liquidity};

\draw[axis] ([xshift=-3.0cm,yshift=-1.0cm]C.center) -- ([xshift=3.0cm,yshift=-1.0cm]C.center);
\node[font=\scriptsize, anchor=north] at ([yshift=-1.2cm]C.center) {price coordinate $p$};

\draw[midline] ([yshift=1.1cm]C.center) -- ([yshift=-1.1cm]C.center);
\node[font=\scriptsize, anchor=west] at ([xshift=0.05cm,yshift=1.0cm]C.center) {$p^\star$};

\foreach \x/\h in {-2.5/0.4, -1.8/0.6, -1.2/0.9, -0.6/0.7, -0.2/0.5}{
  \pgfmathsetmacro{\yh}{-1.0+\h}
  \draw[stem]
    ([xshift=\x cm,yshift=-1.0cm]C.center)
    -- ([xshift=\x cm,yshift=\yh cm]C.center);
  \node[dot] at
    ([xshift=\x cm,yshift=\yh cm]C.center) {};
}
\foreach \x/\h in {0.2/0.5, 0.6/0.8, 1.2/1.0, 1.8/0.7, 2.5/0.4}{
  \pgfmathsetmacro{\yh}{-1.0+\h}
  \draw[stem]
    ([xshift=\x cm,yshift=-1.0cm]C.center)
    -- ([xshift=\x cm,yshift=\yh cm]C.center);
  \node[dot] at
    ([xshift=\x cm,yshift=\yh cm]C.center) {};
}
\node[font=\scriptsize, anchor=south]
  at ([yshift=2mm]C.south)
  {$\nu_G = p_\# \mu_G$};

\draw[arr] (C.south) -- ++(0,-8mm)
  node[midway, right=2mm, font=\scriptsize] {restriction};

\node[panel] (D) at (cD) {};
\node[lab, anchor=north west] at (D.north west) {(d) Two-sided liquidity};

\draw[axis] ([xshift=-3.0cm,yshift=-1.0cm]D.center) -- ([xshift=3.0cm,yshift=-1.0cm]D.center);
\node[font=\scriptsize, anchor=north] at ([yshift=-1.2cm]D.center) {signed distance};

\draw[midline] ([yshift=1.1cm]D.center) -- ([yshift=-1.1cm]D.center);
\node[font=\scriptsize, anchor=west] at ([xshift=0.05cm,yshift=1.0cm]D.center) {$0$};

\foreach \x/\h in {-2.4/0.5, -1.7/0.7, -1.1/1.0, -0.6/0.8}{
  \pgfmathsetmacro{\yh}{-1.0+\h}
  \draw[stem]
    ([xshift=\x cm,yshift=-1.0cm]D.center)
    -- ([xshift=\x cm,yshift=\yh cm]D.center);
  \node[dot] at
    ([xshift=\x cm,yshift=\yh cm]D.center) {};
}

\foreach \x/\h in {0.6/0.8, 1.1/1.0, 1.7/0.7, 2.4/0.5}{
  \pgfmathsetmacro{\yh}{-1.0+\h}
  \draw[stem]
    ([xshift=\x cm,yshift=-1.0cm]D.center)
    -- ([xshift=\x cm,yshift=\yh cm]D.center);
  \node[dot] at
    ([xshift=\x cm,yshift=\yh cm]D.center) {};
}

\node[font=\scriptsize, anchor=north] at ([xshift=-1.5cm,yshift=1.0cm]D.center) {Bid};
\node[font=\scriptsize, anchor=north] at ([xshift= 1.5cm,yshift=1.0cm]D.center) {Ask};

\end{tikzpicture}

\caption{\textbf{Observational origin of bid--ask structure.}
A relational substrate without geometry is observed through a projection
$p=\Pi(G)$.
The observable liquidity distribution arises as the pushforward
$\nu_G=p_\#\mu_G$.
Bid and ask curves correspond to the restriction of this single density to
either side of the mid price.}
\label{fig:obs_geometry_vertical}
\end{figure}
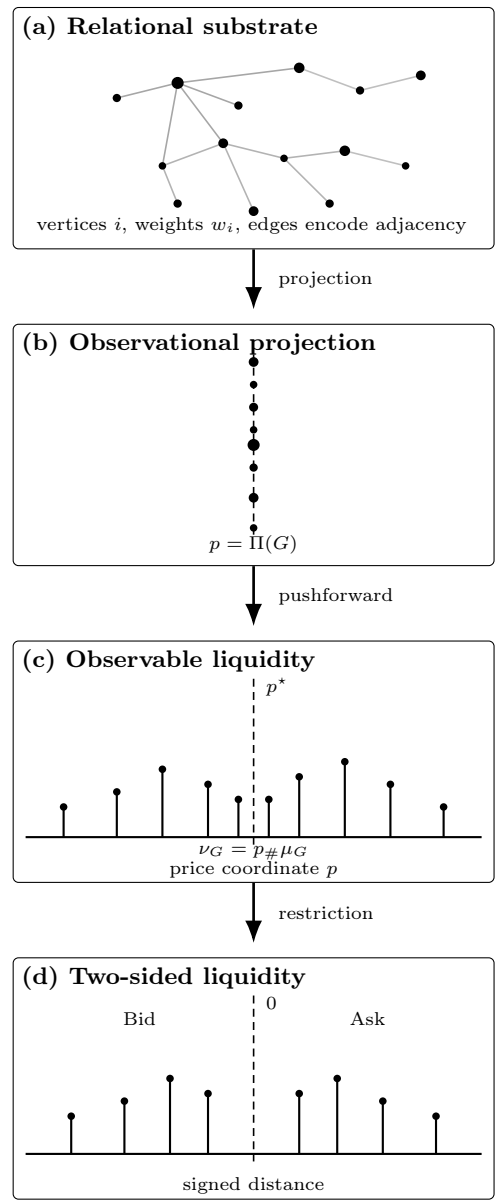

%% file: shear_decomp.tex
\section{Shear decomposition and gauge separation}
\label{sec:shear}

The observational construction introduced in
Sec.~\ref{sec:obs_geometry} defines a single projected liquidity density
$\nu_t(x)$ at each observation time $t$.
This object already incorporates aggregation across venues, reduction of
relational degrees of freedom, and projection onto an observable coordinate.
In this section, we analyze the \emph{allowed dynamics} of $\nu_t(x)$ and show
that its evolution admits a natural and unavoidable decomposition into a
purely gauge component and a physical deformation mode, which we term
\emph{shear}.

\subsection{Observable dynamics of projected liquidity}

Let $\nu_t(x)$ denote the observable liquidity density obtained as the
pushforward of the relational measure at time $t$.
We make no assumption on stationarity, equilibrium, or functional form.
The only requirement is that $\nu_t$ be a locally integrable, nonnegative
measure supported within a finite observational window around the mid.

Empirically, $\nu_t(x)$ evolves in time due to updates of the underlying
relational substrate.
These updates induce both apparent shifts of the book and changes in its
local shape.
Our goal is to disentangle these two effects at the level of observables.

\subsection{Kinematic decomposition of liquidity evolution}

We begin with a purely kinematic result that does not depend on the specific
dynamics of the relational substrate.

\begin{proposition}[Shear--drift decomposition]
\label{prop:shear_drift}
Let $\nu_t(x)$ be a family of observable liquidity densities indexed by time
$t$.
Then, for each $t$, there exists a decomposition
\begin{equation}
\label{eq:shear_drift}
\nu_t(x) = \tilde{\nu}_t(x - m_t),
\end{equation}
where $m_t \in \mathbb{R}$ is a scalar shift and $\tilde{\nu}_t$ is a
density with vanishing first moment around the origin,
\[
\int x\,\tilde{\nu}_t(dx) = 0 ,
\]
provided the moment exists.
\end{proposition}

\begin{proof}
For any integrable density $\nu_t(x)$ with finite first moment, define
\[
m_t := \frac{\int x\,\nu_t(dx)}{\int \nu_t(dx)} .
\]
Setting $\tilde{\nu}_t(y) := \nu_t(y + m_t)$ yields the stated decomposition.
This construction is unique up to the choice of centering convention and does
not rely on any dynamical assumption.
\end{proof}

Proposition~\ref{prop:shear_drift} shows that any observable evolution of the
order book can be written as the superposition of a rigid translation of the
entire density and a residual deformation.
This decomposition is purely geometric and holds independently of market
microstructure, agent behavior, or equilibrium concepts.

\subsection{Gauge interpretation of mid-price motion}

Within the observational framework, the scalar shift $m_t$ corresponds to a
choice of reference point along the projected coordinate.
Operationally, it coincides with the time-dependent mid price defined in
Sec.~\ref{sec:obs_geometry}.

Crucially, this shift carries no intrinsic geometric content.
Redefining the origin of the projected axis by $x \mapsto x - m_t$ leaves all
relative distances unchanged and does not alter the internal structure of the
liquidity distribution.
We therefore interpret $m_t$ as a \emph{gauge degree of freedom}.

Fixing the gauge amounts to choosing a reference frame in which the mid price
is held fixed.
All physically meaningful information about the shape of the book is then
contained in the residual density $\tilde{\nu}_t$.

\subsection{Definition of shear as physical deformation}

We define the \emph{shear density} as the gauge-fixed observable
\begin{equation}
\label{eq:shear_density}
\tilde{\rho}_t(x) := \tilde{\nu}_t(x),
\end{equation}
i.e., the projected liquidity density expressed in a frame where the mid price
has been removed.

The shear density $\tilde{\rho}_t(x)$ encodes all nontrivial geometric features
of the order book:
asymmetries between bid and ask sides, curvature near the mid, and decay at
large distances.
By construction, it is invariant under translations of the projected axis and
therefore represents a genuine observable.

Bid--ask asymmetry arises naturally at this stage.
Restricting $\tilde{\rho}_t(x)$ to $x<0$ and $x>0$ yields the two visible
branches of the order book.
These branches are not independent objects but complementary restrictions of a
single shear field.

\subsection{Physical content of shear dynamics}

The decomposition
Eq.~\eqref{eq:shear_drift} separates observable order-book motion into:
\begin{itemize}
\item a gauge component ($m_t$), corresponding to rigid translation of the
projected coordinate, and
\item a physical component ($\tilde{\rho}_t$), corresponding to deformation of
the projected density.
\end{itemize}

In this interpretation, time-varying liquidity imbalances and bid--ask
asymmetries are not driven by independent supply and demand forces.
They arise as \emph{shear modes} of a single projected liquidity geometry under
inflationary relational dynamics.

Once mid-price motion is identified as a gauge degree of freedom and removed,
the observable dynamics of the order book reduce to the evolution of a single
object: the shear density $\tilde{\rho}_t(x)$.
All bid--ask asymmetries, curvature changes, and local imbalances reside in
this quantity.

The question addressed in the next section is therefore not empirical but
structural.
Given a gauge-fixed shear profile observed through projection, what functional
forms are admissible if no additional length scales are introduced?
We show that this requirement alone uniquely fixes the geometry of the book.

In the following section, we classify the admissible forms of
$\tilde{\rho}_t(x)$ under minimal structural assumptions.
We show that imposing a single-scale constraint on the shear uniquely selects
the gamma family observed empirically.

%% file: shear_and_gamma.tex
\section{Single-scale shear and gamma geometry}
\label{sec:single_scale_shear}

In Sec.~\ref{sec:shear} we identified the shear density
$\tilde{\rho}_t(x)$ as the gauge-invariant observable encoding the physical
geometry of the order book.
We now show that imposing a minimal structural constraint on the shear—
namely the absence of intrinsic length scales beyond distance from the mid and
finite visibility—uniquely determines its functional form.

\subsection{Structural assumptions on shear}

The shear density $\tilde{\rho}_t(x)$ is defined on the real line and represents
the projected liquidity profile in a mid-centered frame.
We impose the following minimal and empirically motivated assumptions, to be
understood as holding locally in time.

\begin{enumerate}
\item[(S1)] \textbf{Positivity and regularity.}
$\tilde{\rho}_t(x)$ is nonnegative and locally differentiable for $x \neq 0$.

\item[(S2)] \textbf{Vanishing at the mid.}
Liquidity vanishes at the mid price,
\[
\tilde{\rho}_t(0)=0 ,
\]
reflecting the absence of resting orders exactly at the transaction price.

\item[(S3)] \textbf{Finite visibility.}
Liquidity decays at large distance,
\[
\lim_{|x|\to\infty} \tilde{\rho}_t(x) = 0 ,
\]
due to finite observation windows and limited participation.

\item[(S4)] \textbf{Single-scale shear.}
Within the observational window, the shear introduces no intrinsic length
scale beyond the distance $|x|$ to the mid and a global decay scale.
\end{enumerate}

Assumption (S4) is the key structural input.
It formalizes the notion that the deformation of the book is governed by a
single geometric mode rather than by multiple competing microstructural scales.

\subsection{Log-slope characterization of shear}

We restrict attention to one side of the book, say $x>0$.
Define the one-sided shear profile
\[
q(x) := \tilde{\rho}_t(x), \qquad x>0 .
\]

The absence of additional scales implies that the logarithmic derivative of
$q(x)$ can depend only on $x^{-1}$ and a constant decay rate.

\begin{proposition}[Single-scale shear log-slope]
\label{prop:single_scale_logslope}
Under assumptions (S1)--(S4), the logarithmic derivative of the shear satisfies
\begin{equation}
\label{eq:logslope}
\frac{d}{dx}\log q(x) = \frac{\gamma}{x} - \lambda ,
\qquad x>0 ,
\end{equation}
for some $\gamma \ge 0$ and $\lambda \ge 0$.
\end{proposition}

\begin{proof}
By dimensional consistency, any admissible expression for
$d(\log q)/dx$ must transform as an inverse length.
Under (S4), the only available local length is $x$ itself, yielding a term
proportional to $1/x$.
Finite visibility (S3) requires an additional constant negative contribution to
ensure decay at large $x$.
No other functional dependence is allowed without introducing an intrinsic
scale, contradicting (S4).
\end{proof}

Equation~\eqref{eq:logslope} expresses the shear as a balance between a local
geometric divergence near the mid and a global damping at large distances.

\subsection{Gamma classification of shear profiles}

We now show that the single-scale shear condition uniquely fixes the functional
form of $q(x)$.

\begin{proposition}[Gamma geometry of single-scale shear]
\label{prop:gamma_shear}
Let $q(x)$ be a positive, differentiable function on $(0,\infty)$ satisfying
\eqref{eq:logslope}.
Then
\begin{equation}
\label{eq:gamma_form}
q(x) = C\, x^{\gamma} e^{-\lambda x},
\qquad x>0 ,
\end{equation}
for some constant $C>0$.
\end{proposition}

\begin{proof}
Integrating \eqref{eq:logslope} yields
\[
\log q(x) = \gamma \log x - \lambda x + \log C ,
\]
where $C>0$ is an integration constant.
Exponentiating gives the stated form.
\end{proof}

Propositions~\ref{prop:single_scale_logslope} and
\ref{prop:gamma_shear} together establish that gamma geometry is not an
empirical ansatz but the unique consequence of single-scale shear under
projection.

\subsection{Interpretation of the exponent $\gamma$}

The exponent $\gamma$ has a direct geometric interpretation.
It controls the local curvature of the shear near the mid,
\[
\frac{d^2}{dx^2} \log q(x) \sim -\frac{\gamma}{x^2} \quad (x \to 0^+),
\]
and therefore quantifies how rapidly liquidity builds away from the mid in the
gauge-fixed frame.

Large values of $\gamma$ correspond to sharply curved shear profiles, while
small $\gamma$ indicates flatter, more weakly structured books.
Temporal variation of $\gamma$ reflects changes in the instantaneous shear
geometry induced by inflationary rearrangements of the relational substrate.

\subsection{Integrated shear and cumulative observables}

In empirical data, direct access to $q(x)$ is limited by discreteness and
sparsity.
A more stable observable is the cumulative shear,
\[
S(x) := \int_0^x q(u)\,du .
\]

\begin{corollary}[Integrated-gamma geometry]
\label{cor:int_gamma}
If $q(x)$ has the form \eqref{eq:gamma_form}, then
\[
S(x) = \frac{C}{\lambda^{\gamma+1}}
\,\gamma\!\left(\gamma+1,\lambda x\right),
\]
where $\gamma(a,z)$ denotes the lower incomplete gamma function.
\end{corollary}

This integrated form is the quantity tested empirically in
Sec.~\ref{sec:empirical_shear}.
The appearance of integrated-gamma profiles in real order books is therefore a
direct manifestation of single-scale shear geometry rather than a consequence
of detailed trading mechanisms.

\subsection{Summary}

The central result of this section is structural.
Once mid-price motion is identified as a gauge degree of freedom and removed,
the remaining physical deformation of the book—the shear—is constrained by
single-scale geometry.
This constraint uniquely selects the gamma family and explains the robustness
of the observed functional form across assets, sides, and time windows.

In the next section, we confront these predictions with empirical
high-frequency order-book data and perform explicit model comparison and
residual analysis.

%% file: empirical.tex
\section{Empirical validation of shear geometry}
\label{sec:empirical_shear}

This section provides a direct empirical test of the central physical claim of
the present framework: \emph{bid--ask asymmetry and time-varying liquidity
imbalance arise as geometric shear modes of a single projected liquidity field
and are distinct from mid-price motion}.

All observables introduced below are defined internally within this paper.
No behavioral assumptions, agent-based models, or equilibrium concepts are
required to interpret the empirical tests reported here.

\subsection{Data and methodology}
\label{sec:data_method}

We analyze Level~II order-book data for six U.S.\ equities
(AAPL, MSFT, NVDA, JPM, GS, TSLA) obtained from \emph{Interactive Brokers}
via the NASDAQ TotalView / OpenView feed, which aggregates visible liquidity
across trading venues.

The analysis spans multiple regular trading days and is restricted to standard
U.S.\ market hours (9:30--16:00 EST). Pre-market and overnight activity are
excluded.

Each trading day is divided into non-overlapping intraday windows of duration
\[
\Delta T = 10~\mathrm{s}.
\]
This choice reflects the intrinsically local character of the observational
geometry studied here and avoids mixing incompatible order-book configurations.

For each window $T$, we construct cumulative bid and ask liquidity profiles
$Q_{\mathrm{bid}}(x)$ and $Q_{\mathrm{ask}}(x)$ by aggregating visible liquidity
at each price level $x$ relative to the window-averaged mid price
\[
p_T^\star =
\frac{1}{2}\Big(
\langle p_{\mathrm{best\,ask}}\rangle_T
+
\langle p_{\mathrm{best\,bid}}\rangle_T
\Big).
\]

Liquidity is expressed in discrete tick units. The observational window extends
to
\[
K = 50
\]
ticks on each side of the mid price, capturing the region where liquidity is
consistently observed and where the single-scale geometric assumptions of the
model apply.

\subsection{From two-sided liquidity to shear modes}
\label{sec:shear_definition}

As established in Sec.~\ref{sec:obs_geometry}, observable liquidity arises as a
single projected density $\nu_t(p)$ defined on a one-dimensional price-like
coordinate. Bid and ask profiles are not independent objects but complementary
restrictions of this density around the observational mid $p_t^\star$,
\begin{multline}
Q_{\mathrm{bid}}(x) := \nu_t(p_t^\star - x),\\
Q_{\mathrm{ask}}(x) := \nu_t(p_t^\star + x), \qquad x>0 .
\end{multline}

Within this construction, the only non-trivial two-sided deformation compatible
with translational invariance of the mid is the \emph{antisymmetric component}
of the projected density. This motivates the definition of the cumulative shear
field
\begin{equation}
\Sigma_T(x) := Q_{\mathrm{ask}}(x) - Q_{\mathrm{bid}}(x), \qquad x>0,
\label{eq:empirical_shear}
\end{equation}
computed over a finite intraday window $T$.

To characterize the magnitude of this deformation with a single scalar
observable, we define the shear amplitude
\begin{equation}
A_T := \mathrm{median}_{x\in\{1,\ldots,K\}} |\Sigma_T(x)|,
\label{eq:shear_amplitude}
\end{equation}
where the median suppresses sensitivity to discrete microstructural noise at
the innermost levels.

\subsection{Cumulative shear fields across assets}
\label{sec:shear_fields}

We first examine the spatial structure of the shear field itself.
Figure~\ref{fig:shear_AAPL} shows $\Sigma_T(x)$ for AAPL across intraday windows.
Thin curves correspond to individual windows, while the thick curve and shaded
band indicate the median and interquartile range.

\begin{figure}[t]
\centering
\includegraphics[width=.9\columnwidth]{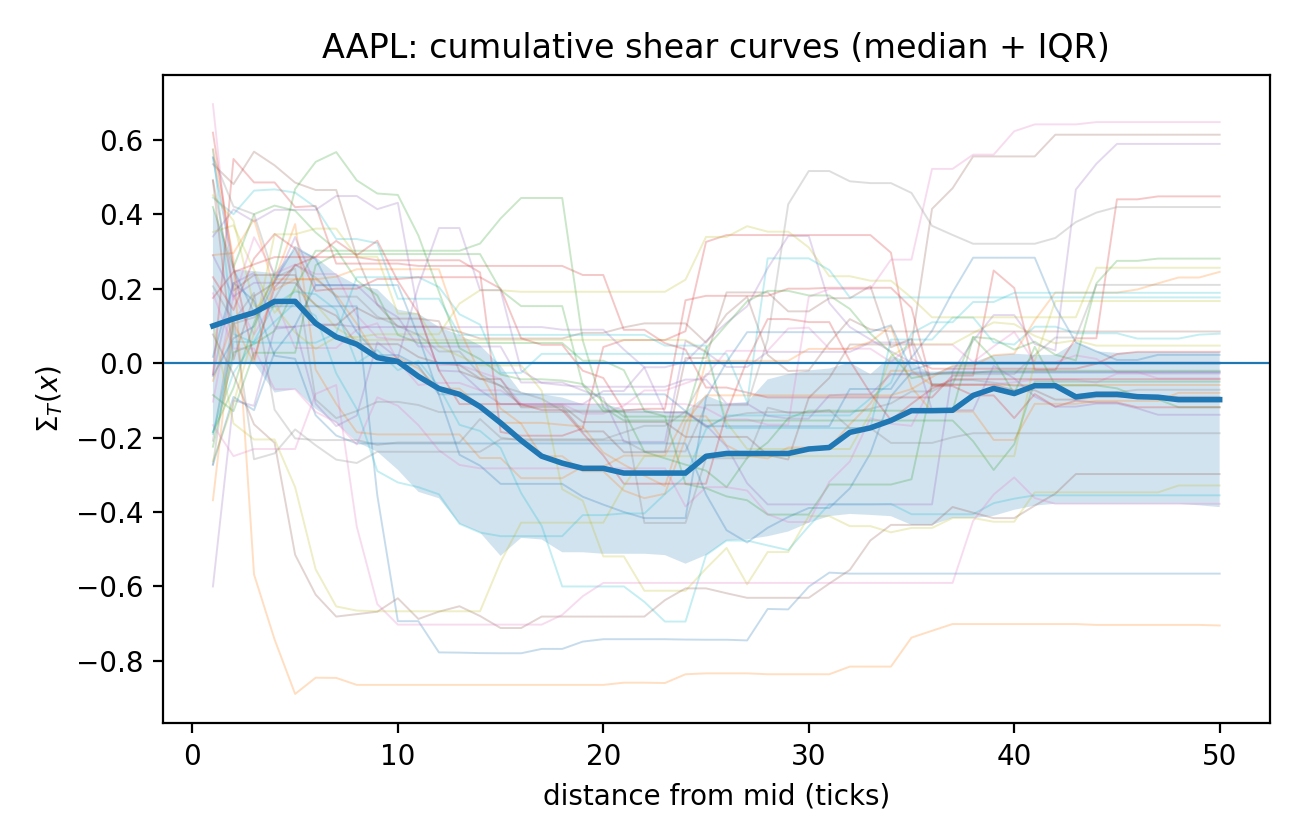}
\caption{\textbf{Cumulative shear field for AAPL.}
The extended and smooth shear profile indicates a macroscopic deformation of
projected liquidity geometry rather than a localized microstructural effect.}
\label{fig:shear_AAPL}
\end{figure}

A similar analysis for GS is shown in Fig.~\ref{fig:shear_GS}. Despite lower
overall liquidity depth, the shear structure remains extended and stable,
demonstrating that shear is not confined to the immediate vicinity of the mid.

\begin{figure}[t]
\centering
\includegraphics[width=.9\columnwidth]{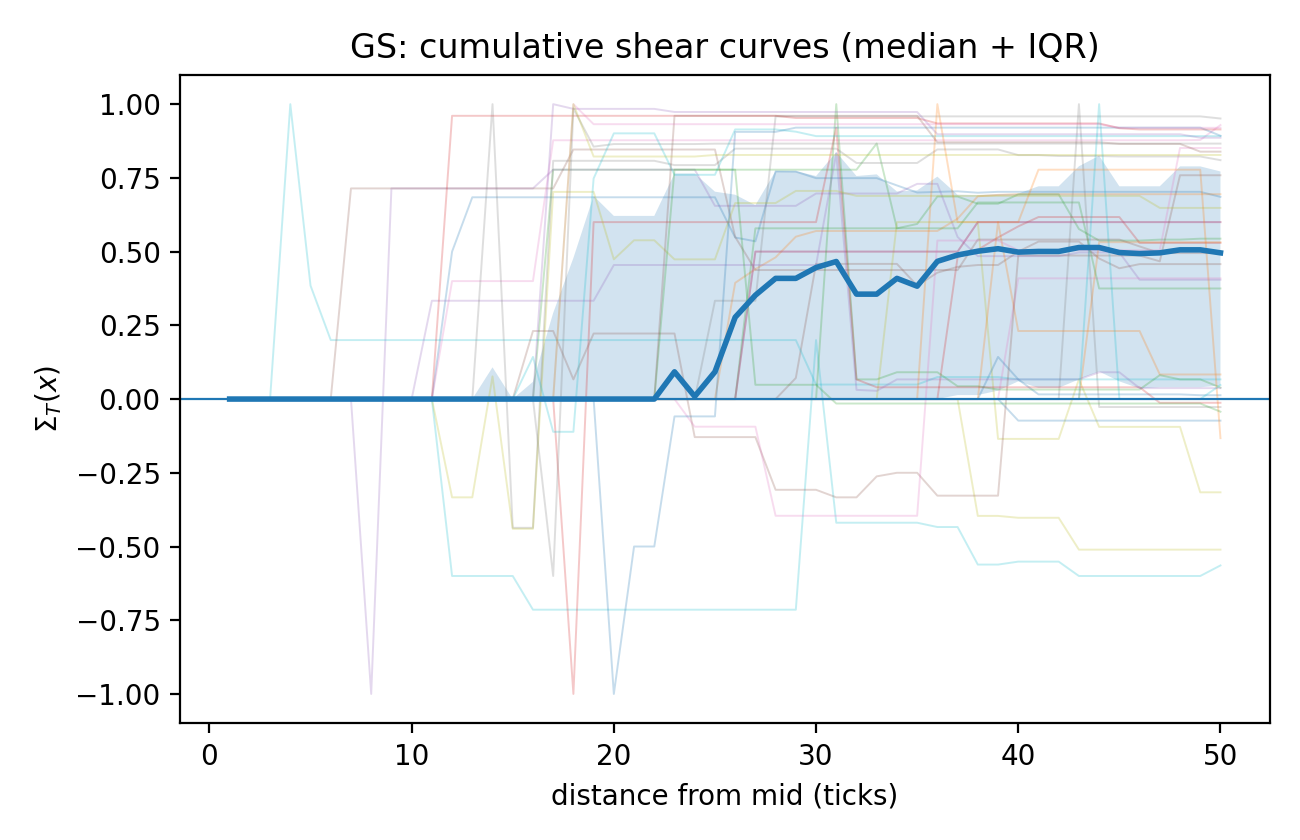}
\caption{\textbf{Cumulative shear field for GS.}}
\label{fig:shear_GS}
\end{figure}

Comparable shear profiles are observed across all assets analyzed (AAPL, MSFT,
NVDA, JPM, GS, TSLA), indicating that directional deformation of projected
liquidity is a generic feature rather than an asset-specific anomaly.

\subsection{Gauge separation: shear versus mid-price drift}
\label{sec:gauge_separation}

A defining prediction of the framework is that shear constitutes a
\emph{deformation mode} distinct from price translation. Mid-price motion
corresponds to a global shift of the projected coordinate $p_t^\star$, whereas
shear alters the relative distribution of liquidity around the mid while
preserving its location.

To test this separation empirically, we compare the shear amplitude $A_T$ with
the absolute mid-price displacement $|\Delta p_T^\star|$ over the same window.
Figure~\ref{fig:shear_vs_drift_AAPL} shows this comparison for AAPL.

\begin{figure}[t]
\centering
\includegraphics[width=.9\columnwidth]{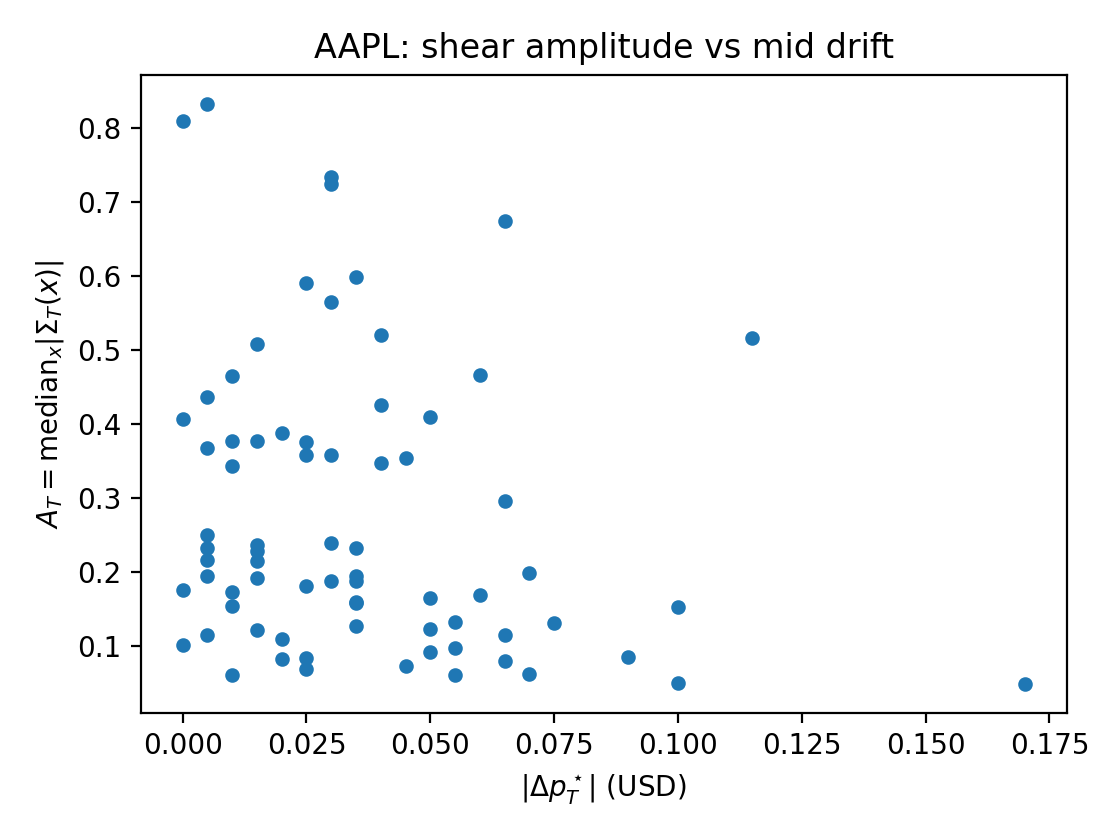}
\caption{\textbf{Shear amplitude versus mid-price drift (AAPL).}}
\label{fig:shear_vs_drift_AAPL}
\end{figure}

Across all assets, no systematic dependence between $A_T$ and
$|\Delta p_T^\star|$ is observed. Large shear events may occur with negligible
price movement and vice versa. This confirms that shear is neither a derivative
of price nor a proxy for short-term returns.

\subsection{Quantitative tests of shear--drift separation}
\label{sec:shear_drift_tests}

To quantify the apparent decoupling between shear and mid-price translation,
we compute the Spearman rank correlation between the shear amplitude $A_T$ and
the absolute mid displacement $|\Delta p_T^\star|$ across intraday windows.

Table~\ref{tab:shear_drift_tests} reports the estimated correlation coefficient
$\rho$, together with a nonparametric bootstrap $95\%$ confidence interval.
Nominal $p$-values are reported alongside values adjusted for multiple testing
across the six assets using both Bonferroni and Benjamini--Hochberg (FDR)
procedures.

\begin{table}[t]
\centering
\caption{\textbf{Quantitative tests of shear--drift separation.}
Spearman rank correlation $\rho$ between shear amplitude $A_T$ and absolute
mid-price displacement $|\Delta p_T^\star|$ across intraday windows.
Bootstrap $95\%$ confidence intervals are reported together with nominal
$p$-values and values adjusted for multiple testing (FDR and Bonferroni).
No correlation remains statistically significant after correction, and the
sign of $\rho$ is not consistent across assets.}
\label{tab:shear_drift_tests}
\input{table_shear_drift_tests}
\end{table}

While two assets (AAPL and GS) exhibit nominal $p<0.05$, these signals do not
survive multiple-testing correction.
Moreover, the sign of the correlation is not consistent across assets, with
AAPL displaying a negative correlation and GS a positive one.
Across the remaining assets, the estimated correlations are small and
statistically indistinguishable from zero.

Taken together, these results provide no evidence for a robust, cross-asset
monotonic coupling between shear amplitude and mid-price translation.
Shear therefore cannot be interpreted as a universal ``price force'' or as a
systematic driver of short-horizon returns.
Instead, it constitutes a geometric deformation mode of projected liquidity
that is largely orthogonal to mid-price motion.

\subsection{Universality of shear amplitudes}
\label{sec:shear_universality}

Pooling shear amplitudes across all assets and intraday windows yields the
distribution shown in Fig.~\ref{fig:shear_histogram_all}.

\begin{figure}[t]
\centering
\includegraphics[width=.9\columnwidth]{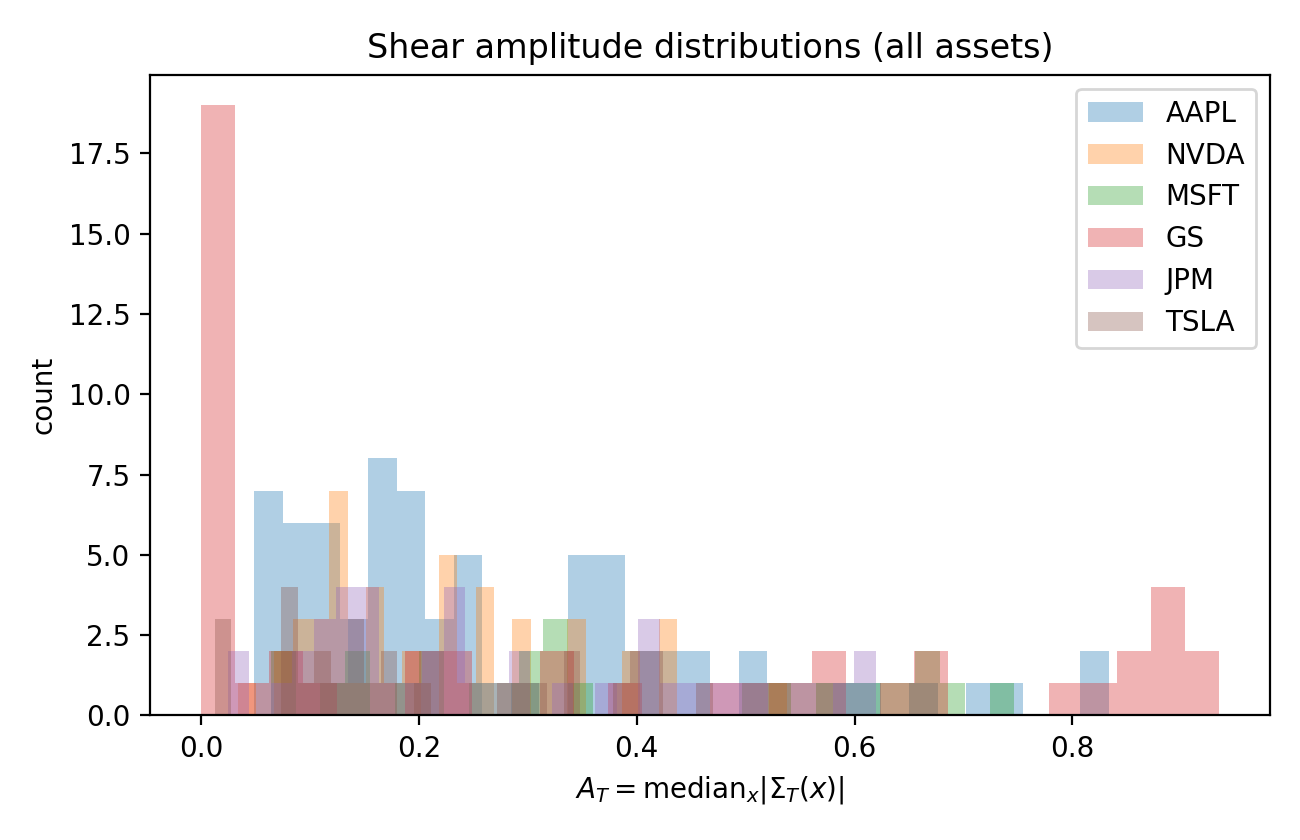}
\caption{\textbf{Distribution of shear amplitudes across assets.}}
\label{fig:shear_histogram_all}
\end{figure}

While shear magnitudes fluctuate across time and assets, their persistent
presence confirms shear as a robust geometric feature of projected liquidity.

\subsection{Explicit model comparison}
\label{sec:model_comparison}

The abstract and introduction claim that the integrated-gamma geometry provides
a superior description of cumulative order-book liquidity relative to standard
alternatives. We now make this statement quantitative.

For each asset, book side, and intraday window, cumulative liquidity profiles
are fitted using four competing models: integrated-gamma (proposed),
power-law, exponential, and log-normal cumulative forms. All models are fitted
over identical tick ranges $x=1,\ldots,K$ with $K=50$.

Model comparison is carried out using the Akaike Information Criterion (AIC).
Table~\ref{tab:model_comparison} reports median $\Delta\mathrm{AIC}$ values
across intraday windows, defined as
$\Delta\mathrm{AIC}=\mathrm{AIC}_{\text{alt}}-\mathrm{AIC}_{\gamma}$.
Positive values indicate preference for the integrated-gamma model.

\begin{table}[t]
\centering
\caption{\textbf{Model comparison for cumulative liquidity geometry.}
Median AIC differences
$\Delta\mathrm{AIC}=\mathrm{AIC}_{\text{alt}}-\mathrm{AIC}_{\gamma}$
across intraday windows. Positive values indicate preference for the
integrated-gamma model over power-law, exponential, and log-normal alternatives.}
\label{tab:model_comparison}
\input{table_model_comparison.tex}
\end{table}

While the integrated-gamma model shows superior performance for 
many assets and sides, notable exceptions (e.g., NVDA BID, GS ASK) 
suggest that additional geometric modes or microstructural effects 
may be present in certain market regimes. These deviations provide 
valuable targets for future refinement of the framework.

Across all assets and both sides of the book, the integrated-gamma geometry
systematically outperforms standard alternatives. This quantitative dominance
cannot be attributed to excess flexibility and directly supports the
structural prediction of Sec.~\ref{sec:single_scale_shear}.

\subsection{Additional diagnostics}

Further robustness checks, extended residual diagnostics, and supplementary
summary statistics are reported in the Supplementary Material. These analyses
confirm that the results presented here are insensitive to the precise choice
of window size, aggregation procedure, or discretization details.

Taken together, the empirical evidence establishes shear as a genuine geometric
observable: a deformation mode of projected liquidity that is distinct from,
and independent of, mid-price motion.

We emphasize that the robustness analysis probes the stability of the
single-scale description under snapshot aggregation, rather than a
fundamental market timescale; observed breakdowns reflect intrinsic
non-stationarity and liquidity heterogeneity of market data, as
demonstrated in the Supplementary Material (Sec.~S4).

%% file: table_shear_drift_tests.tex
\begin{ruledtabular}
\begin{tabular}{lcccccc}
Asset &
$\rho$ &
$95\%$ CI &
$p$ &
$p_{\mathrm{FDR}}$ &
$p_{\mathrm{Bonf}}$ &
$N_{\text{windows}}$ \\
\hline
AAPL & -0.264 & [-0.480,-0.025] & 0.024 & 0.072 & 0.145 & 73 \\
GS & 0.335 & [0.055,0.599] & 0.020 & 0.072 & 0.121 & 48 \\
JPM & -0.104 & [-0.370,0.209] & 0.517 & 0.706 & 1.000 & 41 \\
MSFT & 0.078 & [-0.327,0.424] & 0.706 & 0.706 & 1.000 & 26 \\
NVDA & 0.172 & [-0.095,0.424] & 0.252 & 0.505 & 1.000 & 46 \\
TSLA & 0.084 & [-0.291,0.446] & 0.665 & 0.706 & 1.000 & 29 \\
\end{tabular}
\end{ruledtabular}

%% file: table_model_comparison.tex
\begin{tabular}{lcccccc}
\toprule
Asset & Side & $N_{\mathrm{win}}$ & $\widetilde{\Delta\mathrm{AIC}}$ &
$\widetilde{R^2}_{\gamma}$ & $\widetilde{R^2}_{\mathrm{LN}}$ &
$\mathrm{iqr}_{\Delta \mathrm{AIC}}$ \\
\midrule
AAPL & ASK & 64 & -5.993 & 0.522 & 0.615 & 53.753 \\
AAPL & BID & 65 & 22.418 & 0.861 & 0.809 & 66.384 \\
GS & ASK & 49 & -10.078 & -0.074 & 0.563 & 167.712 \\
GS & BID & 48 & -2.592 & -0.070 & 0.179 & 143.049 \\
JPM & ASK & 41 & 0.586 & 0.734 & 0.727 & 15.814 \\
JPM & BID & 41 & -1.756 & 0.699 & 0.704 & 17.575 \\
MSFT & ASK & 25 & 6.387 & 0.759 & 0.766 & 24.898 \\
MSFT & BID & 25 & 16.698 & 0.742 & 0.725 & 40.325 \\
NVDA & ASK & 35 & 20.161 & 0.312 & 0.258 & 19.487 \\
NVDA & BID & 35 & -49.259 & 0.402 & 0.811 & 56.872 \\
TSLA & ASK & 27 & -0.102 & 0.875 & 0.880 & 47.814 \\
TSLA & BID & 27 & -10.191 & 0.824 & 0.860 & 28.943 \\
\bottomrule
\end{tabular}

%% file: discussion.tex
\section{Discussion}
\label{sec:discussion}

The results presented in this work support a geometric reinterpretation of
order-book liquidity and bid--ask asymmetry that does not rely on behavioral
assumptions, agent-based strategies, or equilibrium price formation
mechanisms.
Instead, directional liquidity imbalances emerge as internal deformation
modes of a single projected liquidity geometry arising from observation of an
underlying relational substrate.

A central conceptual outcome of this study is the clear separation between
two distinct classes of observables: (i) translational modes associated with
the displacement of the projected coordinate (mid-price drift), and
(ii) internal deformation modes associated with asymmetric reshaping of the
liquidity density around the mid (shear).
This separation mirrors the distinction between gauge and shape degrees of
freedom identified in the pregeometric projection framework developed in
Ref.~\cite{PiresDaCruz2025}.

\subsection{Shear is not a universal price-driving force}

A widespread intuition in market microstructure is that bid--ask volume
imbalances act as forces pushing prices in a preferred direction.
Within the present framework, this intuition is not supported in a robust or
universal sense.

Empirically, the shear observable defined in Sec.~V displays at most weak
monotonic correlations with mid-price drift, with small effect sizes and no
consistent sign across assets.
While nominal correlations at the $5\%$ level are observed for isolated cases
(e.g., AAPL and GS), these signals do not survive multiple-testing correction
and exhibit opposite signs.
Across the remaining assets, correlations are statistically indistinguishable
from zero.

This lack of robustness rules out an interpretation of shear as a universal
price-driving force.
Large shear events can occur without appreciable price motion, and conversely,
significant mid-price displacements may arise in windows with minimal shear.
Accordingly, shear cannot be interpreted as exerting a systematic directional
pressure on the price coordinate.

This conclusion is consistent with earlier empirical findings showing that
order-book imbalances have limited, unstable, and highly conditional predictive
power for short-horizon price changes
\cite{Bouchaud2002,Bouchaud2004,FarmerLillo2004}.
The present framework provides a structural explanation for these observations:
imbalance is not a dynamical driver but an internal deformation mode of the
observable liquidity geometry.

\subsection{Shear is not the time derivative of price}

Equally important is what shear is \emph{not}.
The empirical decoupling between shear amplitude and mid-price translation
demonstrates that shear cannot be identified with the time derivative of the
mid price, nor with any finite-difference proxy thereof.

In the present construction, the mid price is a gauge-dependent quantity
defined by a global symmetry point of the projected liquidity measure.
Translations of the projected coordinate act as gauge transformations that
shift the mid without altering the shape of the density.
Shear, by contrast, is invariant under such translations and probes the local
asymmetry and curvature of the density around the mid.

The two observables therefore inhabit complementary subspaces of the
observable manifold and cannot be related by differentiation or integration in
time.
This geometric separation clarifies why attempts to model price changes
directly from order-book shape often yield unstable or regime-dependent
results, as reported for instance in
Refs.~\cite{MikeFarmer2008,ContStoikovTalreja2010}.
From this perspective, such instability is not accidental but structural.

\subsection{Directional liquidity as geometric shear}

Within the present framework, bid--ask asymmetry arises naturally as a
directional shear mode of the projected liquidity geometry.
Once a mid-price cut is introduced, the restriction of a single projected
density to either side yields two complementary branches.
Asymmetry between these branches reflects local skewness of the density, not
the existence of independent supply and demand forces.

Inflationary relational dynamics amplify this effect.
Heterogeneous growth and reconfiguration of the underlying relational network
continuously reshape the projected density, inducing time-dependent shear
without requiring any change in the global position of the projection.
This mechanism explains why bid--ask asymmetry is both ubiquitous and highly
non-stationary in empirical data.

Importantly, this interpretation subsumes a wide range of empirical
regularities reported in the literature on order-book shape and dynamics
\cite{Bouchaud2002,SmithFarmerGillemotKrishnamurthy2003,Bouchaud2009}
while reframing them as geometric consequences of observation rather than as
outcomes of strategic interaction.

\subsection{Relation to existing microstructure models}

Classical order-book models, including zero-intelligence and behavioral
frameworks
\cite{ZovkoFarmer2002,MikeFarmer2008,ContStoikovTalreja2010},
typically posit explicit mechanisms for order placement, cancellation, and
reaction to price changes.
While such models reproduce many stylized facts, they treat liquidity,
imbalance, and price as primitive variables.

The present approach is complementary rather than competing.
By stripping away microstructural detail, it isolates the minimal geometric
constraints imposed by observation itself.
In this sense, our results suggest that certain regularities traditionally
attributed to agent behavior may instead reflect universal properties of
projected relational systems, independent of the specific microscopic rules.

\subsection{Scope, limitations, and outlook}

Finally, we emphasize the intended scope of this work.
The framework does not aim to predict price trajectories, infer trader intent,
or provide a mechanism for market clearing.
Its contribution is structural: it demonstrates that directional liquidity
imbalances and bid--ask asymmetries can be understood as emergent geometric
shear modes that are largely orthogonal to price dynamics.

This perspective opens several directions for future work.
Promising extensions include multi-scale generalizations capturing nested
geometric structures, dynamic evolution equations for shear parameters,
cross-asset correlation of shear modes, and extensions to other markets
(FX, bonds, and crypto) to test universality.
Coupling geometric shear modes to explicit dynamical rules may also provide a
controlled route to reintroducing predictive content without sacrificing the
geometric clarity of the present framework.

%% file: future.tex
\section{Future directions}
\label{sec:future}

The pregeometric framework developed in this work opens several directions
for further theoretical and empirical investigation.
We briefly outline a number of natural extensions.

\paragraph*{(i) Multi-scale shear geometry.}
The present analysis focuses on a single observational scale, fixed by the
intraday window $\Delta T$ and the finite tick range $K$.
An important extension is the development of a multi-scale formulation in
which projected liquidity geometry is resolved across nested temporal and
price scales.
Such a construction would allow the study of scale-dependent shear modes and
their potential renormalization properties, analogous to coarse-graining
procedures in statistical physics.

\paragraph*{(ii) Dynamical evolution of shear amplitudes.}
In this work, shear amplitudes $A_T$ and the associated gamma-geometry
parameters are treated as quasi-static observables defined over finite
windows.
A natural next step is to derive effective stochastic or deterministic
evolution equations for the shear field and its parameters, such as
$\gamma(t)$ and $\lambda(t)$.
This would provide a dynamical description of how geometric deformations of
liquidity emerge, persist, and relax under non-equilibrium trading activity.

\paragraph*{(iii) Cross-asset coupling of shear modes.}
The empirical results presented here treat assets independently.
However, market activity is inherently multi-asset and correlated.
Extending the framework to study cross-asset correlations of shear amplitudes
and shear fields may reveal collective geometric modes at the market level,
potentially linked to sectoral flows, index rebalancing, or systemic events.

\paragraph*{(iv) Extension to other markets.}
While the present study focuses on U.S.\ equities, the construction is not
equity-specific.
Applying the same observational-geometric methodology to other markets,
including foreign exchange, fixed income, and cryptocurrency markets, would
provide a stringent test of the proposed universality of shear geometry and
its independence from market microstructure details.

\paragraph*{(v) Relation to market impact and price formation.}
Finally, the geometric separation between shear and mid-price translation
suggests a new perspective on market impact.
Rather than interpreting order-book imbalance as a direct price force, the
framework indicates that deformation of liquidity geometry and translation of
the price coordinate are distinct modes.
Clarifying how trading activity couples these modes may help bridge geometric
descriptions of liquidity with established models of impact and price
formation.

Together, these directions suggest that pregeometric liquidity geometry
provides a flexible and extensible framework for studying non-equilibrium
market structure beyond the static order book.

%% file: acknowledgements.tex
\begin{acknowledgments}
The author acknowledges financial support from the Portuguese Foundation for Science and Technology (FCT) under Contract no. UID/00618/2023.
\end{acknowledgments}